\documentclass[12pt]{article}
\usepackage{amsmath, amsthm, amscd, amsfonts, amssymb, graphicx, color}

\newtheorem{theorem}{Theorem}
\newtheorem{corollary}[theorem]{Corollary}

\def\title#1{{\Large\bf  \begin{center} #1 \vspace{0pt} \end{center}  } }
\def\authors#1{{\large\bf \begin{center} #1 \vspace{0pt} \end{center} } }
\def\university#1{{\sl \begin{center} #1 \vspace{0pt} \end{center} } }
\def\inst#1{\unskip $^{#1}$}
\newcommand{\vertiii}[1]{{\left\vert\kern-0.25ex\left\vert\kern-0.25ex\left\vert 
#1 \right\vert\kern-0.25ex\right\vert\kern-0.25ex\right\vert}}
\usepackage{amsfonts}

\def \argmin {\mbox{\boldmath $\mbox{argmin}$}}
\def \Im   {\text {\rm Im}}

\def \tr   {\text {\rm tr}}
\def \Re   {\text {\rm Re}}
\def \det   {\text {\rm det}}
\def \diag  {\text {\rm diag}}
\def \Eig  {\text {\rm Eig}}

\begin{document}

%
%
%

\title{On symplectic eigenvalues of positive definite matrices}

\bigskip

%
%

\authors{Rajendra Bhatia\inst{1}, 
Tanvi Jain\inst{2} 
}


\smallskip

%
%

\university{\inst{1} Indian Statistical Institute, New Delhi 110016, India\\rbh@isid.ac.in}
\university{\inst{2} Indian Statistical Institute, New Delhi 110016, 
India\\tanvi@isid.ac.in}%

\begin{abstract}
If $A$ is a $2n \times 2n$ real positive definite matrix, then there exists a 
symplectic matrix $M$ such that $M^TAM = \left [ \begin{array}{cc} D & O \\ O & 
D \end{array} \right ]$ where $D= \diag (d_1 (A), \ldots, d_n(A))$ is 
a diagonal matrix with positive diagonal entries, which are called the 
symplectic eigenvalues of $A.$ In this paper we derive several fundamental 
inequalities about these numbers. Among them are relations between the 
symplectic eigenvalues of $A$ and those of  $A^t,$ between the symplectic 
eigenvalues of $m$ matrices $A_1, \ldots, A_m$ and of their Riemannian mean, a 
perturbation theorem, some variational  principles, and some inequalities 
between the symplectic and ordinary eigenvalues.
\end{abstract}


%
%
%

\section{Introduction} Let $\mathbb{M}(2n)$ be the space of $2n \times 2n$ real 
matrices, $\mathbb{P}(2n)$ the subset of $\mathbb{M}(2n)$ consisting of 
positive definite matrices, and $Sp(2n)$ the group of real symplectic matrices; 
i.e., 
$$Sp(2n) = \left \{ M \in \mathbb{M}(2n) : M^T JM = J \right \}.$$
Here $J = \left [\begin{array}{cc} O & I \\ -I & O \end{array} \right ],$ and 
$J$ itself is a symplectic matrix.

If $A$ is an element of $\mathbb{P} (2n),$ then there exists a symplectic 
matrix $M$ such that
\begin{equation}
M^T AM = \left [\begin{array}{cc} D & O \\ O & D \end{array} \right ], 
\label{eq1}
\end{equation}
where $D$ is a diagonal matrix with positive entries
\begin{equation}
 d_1 (A) \leq d_2 (A) \leq \cdots \leq d_n (A). \label{eq2}
\end{equation}
This is often called Williamson's theorem \cite{ar}, \cite{g}, \cite{p}. In 
\cite{hz} it is pointed out that this was known to Weierstrass. The numbers 
$d_j(A)$ are uniquely determined by $A$ and characterise the orbits of 
$\mathbb{P}(2n)$ under the action of the group $Sp (2n).$ We call them the {\it 
symplectic eigenvalues} of $A.$  They play an important role in classical 
Hamiltonian dynamics \cite{ar}, in quantum mechanics \cite{dms}, in symplectic topology \cite{hz}, and in the 
more recent subject of quantum information; see e.g., \cite{e}, \cite{h}, \cite{ktv}, 
\cite{p}.

The goal of this paper is to present some fundamental inequalities for 
symplectic eigenvalues.

It is clear from the definition that if the symplectic eigenvalues of $A$ are 
enumerated as in \eqref{eq2}, then those of $A^{-1}$ are
\begin{equation}
\frac{1}{d_n(A)} \leq \frac{1}{d_{n-1} (A)} \leq \cdots \leq  \frac{1}{d_1(A)}. 
 \label{eq3}
\end{equation}
No relation between the symplectic eigenvalues of $A$ and those of $A^t$ is 
readily apparent. Our first theorem unveils such relationships.

Given $x \in \mathbb{R}^m_{+},$ we denote by $x^{\downarrow} = \left 
(x_1^{\downarrow}, \ldots, x_m^{\downarrow} \right )$ the vector whose 
coordinates are the coordinates of $x$ rearranged in decreasing order 
$x_1^{\downarrow} \ge \cdots \ge x_m^{\downarrow}.$ If $x$ and $y$ are two 
$m$-vectors with positive coordinates, then we say that $x$ is {\it log 
majorised} by $y,$ in symbols  $x \prec_{\log} y,$ if
\begin{equation}
\prod_{j=1}^{k} x_j^{\downarrow} \leq \prod_{j=1}^{k} y_j^{\downarrow}, \quad 1 
\leq k \leq m   \label{eq4}
\end{equation}
and 
\begin{equation}
\prod_{j=1}^{m} x_j^{\downarrow} = \prod_{j=1}^{m} y_j^{\downarrow}.   
\label{eq5}
\end{equation}
By classical theorems of Weyl and Polya, $\log$ majorisation implies the 
usual {\it weak majorisation} relation $x \prec_{w} y$ characterised by the 
inequalities
\begin{equation}
\sum_{j=1}^{k} x_j^{\downarrow} \leq \sum_{j=1}^{k} y_j^{\downarrow}, \quad 1 
\leq k \leq m.         \label{eq6}
\end{equation}
See Chapter II of \cite{rbh}.

It is convenient to introduce a $2n$-vector $\widehat{d}(A)$ whose coordinates 
are
\begin{equation}
\widehat{d}_1 (A) \ge  \widehat{d}_2 (A) \ge \cdots \ge   \widehat{d}_{2n} (A), 
\label{eq7}
\end{equation}
which are the symplectic eigenvalues of $A,$ each counted twice and rearranged 
in 
{\it decreasing} order. (Thus $\widehat{d}_1 (A) = \widehat{d}_2 (A) = d_n (A)$ and 
$\widehat{d}_{2n-1} (A) = \widehat{d}_{2n} (A) = d_1 (A).$) With
these notations we have the following.

\begin{theorem}\label{thm1}
Let $A$ be any element of $\mathbb{P}(2n).$ Then
\begin{equation}
 \widehat{d} (A^t) \prec_{\log}  \widehat{d\,\,}^t (A) \quad \mbox{for}\quad 0 
\leq 
t \leq 1, \label{eq8}
\end{equation}
and 
\begin{equation}
   \widehat{d\,\,}^t (A)   \prec_{\log}  \widehat{d} (A^t)\quad \mbox{for}\quad 
1 
\leq t < \infty.  \label{eq9}
\end{equation}
\end{theorem}

\begin{corollary}\label{cor2}
 The symplectic eigenvalues of $A$ have the properties:
\begin{itemize}
 \item[(i)] If $0 \leq t \leq 1,$ then for all $1 \leq k \leq n$
\begin{equation}
 \prod_{j=1}^{k} d_j (A^t) \ge \prod_{j=1}^{k} d_j^t (A).   \label{eq10}
\end{equation}
\item[(ii)] If $t \ge 1,$ then for all $1 \leq k \leq n$
\begin{equation}
\prod_{j=1}^{k} d_j (A^t) \leq \prod_{j=1}^{k} d_j^t (A).    \label{eq11}
\end{equation}
\end{itemize}
\end{corollary}

Given two $n \times n$ positive definite matrices $A$ and $B,$ their {\it 
geometric mean} $G(A, B),$ also denoted as $A \# B,$ is defined as
\begin{equation}
G (A, B) = A \# B = A^{1/2} (A^{-1/2} B A^{-1/2})^{1/2} A^{1/2}.  \label{eq12}
\end{equation}
This was introduced by Pusz and Woronowicz \cite{pw}, and has been much studied 
in connection with problems in physics, electrical networks, and matrix 
analysis. Recently there has been renewed interest in it because of its 
interpretation as the midpoint of the geodesic joining $A$ and $B$ in the 
Riemannian manifold $\mathbb{P}(n).$ The Riemannian distance between $A$ and 
$B$ is defined as
\begin{equation}
\delta (A,B) = \left ( \sum_{i=1}^{n} \,\log^2 \,\lambda_i \,(A^{-1} B) \right 
)^{1/2},    \label{eq13}
\end{equation}
where $\lambda_i(X),$ $1 \leq i \leq n,$ are the eigenvalues of $X.$ With this 
metric $\mathbb{P}(n)$ is a nonpositively curved space. Any two points $A$ and 
$B$ in $\mathbb{P}(n)$ can be joined by a unique geodesic. A natural 
parametrisation for this geodesic is
\begin{equation}
A \#_t B = A^{1/2} \left (A^{-1/2} B A^{-1/2} \right )^t A^{1/2}, \quad 0 \leq 
t \leq 1.    \label{eq14}
\end{equation}
$G(A,B)$ is evidently the midpoint of this geodesic. Our next theorem links the 
symplectic eigenvalues of $A \#_t B$ with those of $A$ and $B.$

\begin{theorem}\label{thm3}
 Let $A,B$ be any two elements of $\mathbb{P}(2n).$ Then for $0 \leq t \leq 1,$
\begin{equation}
\widehat{d} \left ( A \#_t B \right ) \prec_{\log} \widehat{d\,\,}^{1-t} (A) 
\widehat{d\,\,}^t (B).    \label{eq15}
\end{equation}
In particular
\begin{equation}
\widehat{d} (G(A,B)) \prec_{\log} \left (\widehat{d} (A) \widehat{d} (B) \right 
)^{1/2}.    \label{eq16}
\end{equation}
\end{theorem}

Next let $A_1, A_2, \ldots, A_m$ be $m$ points in $\mathbb{P}(n).$ Their 
geomeric mean, variously called the {\it Riemannian mean}, the {\it Cartan 
mean}, the {\it Karcher mean}, the {\it Riemannian barycentre,} is defined as
\begin{equation}
G(A_1, \ldots, A_m) = \argmin \sum_{j=1}^{m} \frac{1}{m} \delta^2 (A_j, X).     
\label{eq17} 
\end{equation}
This object of classical differential geometry has received much attention from 
operator theorists and matrix analysts in the past ten years, and many new 
properties of it have been established. It has also found applications in 
diverse areas such as statistics, machine learning, image processing, 
brain-computer interface, etc. We refer the reader to \cite{rbh1} for a basic 
introduction to this area and to \cite{n} for an update.

A little more generally, a weighted geometric mean of $A_1, \ldots, A_m$ can be 
defined as follows. Given positive numbers $w_1,\dots, w_m$ with $\sum w_j=1,$ 
let
\begin{equation}
G (w, A_1, \ldots, A_m) = \argmin \sum_{j=1}^{m} w_j \delta^2 (A_j, X).   
\label{eq18} 
\end{equation}
The object in \eqref{eq17} is the special case when $w_j = \frac{1}{m}$ for all 
$j.$ In case $m=2,$ we have $(w_1, w_2)=(1-t, t)$ for some $0 \leq t \leq 1,$ 
and then $G(w, A, B)$ reduces to the matrix in \eqref{eq14}. With $t=1/2$ it 
reduces further to \eqref{eq12}.

Our next theorem is a several-variables version of Theorem \ref{thm3}.

\begin{theorem}\label{thm4}
 Let $A_1, \ldots, A_m$ be  elements of $\mathbb{P}(2n)$ and let $w = (w_1, 
\ldots, w_m)$ be a positive vector with $\sum w_j = 1.$ Then 
\begin{equation}
\widehat{d} (G(w, A_1, \ldots, A_m)) \prec_{\log} \prod_{j=1}^{m} 
\widehat{d\,\,}^{w_{j}} (A_j). \label{eq19}
\end{equation}
In particular
\begin{equation}
\widehat{d} (G(A_1, \ldots, A_m)) \prec_{\log}     \left ( 
\prod_{j=1}^{m} \widehat{d} (A_j)\right )^{1/m}       \label{eq20}
\end{equation}
\end{theorem}

In the study of eigenvalues of Hermitian matrices, a very important role is 
played by variational principles, such as the Courant-Fischer-Weyl minmax 
principle, Cauchy's interlacing theorem and Ky Fan's  theorems on extremal
characterisations of sums and products of eigenvalues. It will be valuable to 
assemble a similar arsenal of techniques for symplectic eigenvalues. In Section 
4 of this paper we give an exposition of some of these ideas. We provide an 
outline of proofs of a minmax principle and an interlacing theorem (both of 
which are known results). Then we use this to provide a unified simple proof of 
the following theorem. To emphasize the dependence on $n$ we use the notation 
$J_{2n}$ for the $2n \times 2n$ matrix $\left [\begin{array}{cc} O & I \\ -I & 
O \end{array} \right ].$ The minimum in Theorem \ref{thm5} below is taken over 
$2n \times 2k$ matrices $M$ satisfying $M^T J_{2n} M = J_{2k}.$

\begin{theorem}\label{thm5}
 Let $A \in \mathbb{P}(2n).$ Then for all $1 \leq k \leq n$
\begin{eqnarray}
{\rm (i)} & & 2 \sum\limits_{j=1}^{k} d_j (A) = \underset{M:M^{T} J_{2n} 
M= J_{2k}}{\min} \,\tr \, M^T AM, \label{eq21}
\end{eqnarray}
\begin{eqnarray}
{\rm (ii)} & & \prod_{j=1}^{k} \,d_j^2 (A) = \underset{M:M^{T} J_{2n} M= 
J_{2k}}{\min} \, \det \, M^T AM.\label{eq22}
\end{eqnarray}

\end{theorem}
Part (i) of this theorem has been proved by Hiroshima \cite{h}, and was an 
inspiration for our work. Our proof might be simpler and more conceptual. An 
interesting property of symplectic matrices crops up as a byproduct of our 
analysis.

Every element $M$ of $Sp(2n)$ has a block decomposition
\begin{equation}
 M = \left [ \begin{array}{cc} A & B \\ C & G \end{array} \right ], \label{eq23}
\end{equation}
in which $A,B,C,G$ are $n \times n$ matrices satisfying the conditions
\begin{equation}
AG^T - BC^T = I, \,\,AB^T-BA^T=0, \,\,CG^T-GC^T=0.   \label{eq24}
\end{equation}
We associate with $M$ an $n \times n$ matrix $\widetilde{M}$ whose entries are 
given by
\begin{equation}
\widetilde{m}_{ij} = \frac{1}{2} \left (a_{ij}^2 + b_{ij}^2 + c_{ij}^2 + 
g_{ij}^2 \right ). \label{eq25} 
\end{equation}
 This matrix has some nice properties and can be put to good use in the study 
of symplectic matrices. In the course of our proof of Theorem \ref{thm5} we 
will see that for every $M \in Sp (2n)$ the matrix $\widetilde{M}$ has the 
properties
\begin{eqnarray}
 \sum_{j=1}^{n} \widetilde{m}_{ij} \ge 1, && 1 \leq i \leq n, \quad \mbox{and} 
\nonumber\\
\sum_{i=1}^{n} \widetilde{m}_{ij} \ge 1, && 1 \leq j \leq n.  \label{eq26}
\end{eqnarray}
It turns out that more is true.

An $n \times n$ matrix $A$ is said to be {\it doubly stochastic} if $a_{ij} \ge 
0$ for all $i,j,$ 
$$\sum_{j=1}^{n} a_{ij} = 1 \quad \mbox{for all} \quad 1 \leq i \leq n$$
and 
$$\sum_{i=1}^{n} a_{ij} = 1 \quad \mbox{for all} \quad 1 \leq j \leq n.$$
A matrix $B$ with nonnegative entries is called {\it doubly superstochastic} if 
there exists a doubly stochastic matrix $A$ such that $b_{ij} \ge a_{ij}$ for 
all $i,j.$ Our next theorem shows that $\widetilde{M}$ is a doubly 
superstochastic matrix.

\begin{theorem}\label{thm6}
 Let $M \in Sp (2n),$ and let $\widetilde{M}$ be the $n \times n$ matrix 
associated with $M$ according to the rule \eqref{eq25}. Then $\widetilde{M}$ is 
doubly superstochastic. Further $\widetilde{M}$ is doubly stochastic if and 
only if $M$ is orthogonal. 
\end{theorem}

Doubly stochastic, superstochastic and substochastic matrices play an important 
role in the theory of inequalities; see the monograph \cite{moa}. Theorem 
\ref{thm6} is thus likely to be very useful in deriving inequalities for 
symplectic matrices.

For the usual eigenvalues of Hermitian matrices there are several perturbation 
bounds available. See \cite{rbh}. Our next theorem gives such inequalities for 
symplectic eigenvalues. The continuity implied by these bounds will be used in 
our proofs of Theorems \ref{thm1}, \ref{thm3}, \ref{thm4}. But they are of 
independent interest.

We use the symbol $\vertiii{\cdot}$ to denote any {\it unitarily invariant 
norm} on the space of matrices \cite{rbh}. Particular examples are the {\it 
operator norm}
\begin{equation}
\|A\| = \lambda_1 (A^T A)^{1/2} = \underset{\|x\|=1}{\sup}\,\|Ax\|,    
\label{eq27}
\end{equation}
and the {\it Frobenius norm}
\begin{equation}
\|A\|_2 = \left ( \tr \,A^T A \right )^{1/2} = \left ( \sum |a_{ij}|^2 \right 
)^{1/2}.     \label{eq28} 
\end{equation}
Here $\lambda_1$ stands for the maximum eigenvalue.

\begin{theorem}\label{thm7}
Let $A,B$ be two elements of $\mathbb{P}(2n),$ and let $\widehat{D}(A),$ 
$\widehat{D}(B)$ be the diagonal matrices whose diagonals are $\widehat{d}(A)$ 
and $\widehat{d}(B).$ Then for every unitarily invariant norm we have
\begin{equation}
\vertiii{\widehat{D}(A)-\widehat{D}(B)} \leq \left ( \|A\|^{1/2} + \|B\|^{1/2} 
\right ) \vertiii{|A-B|^{1/2}}. \label{eq29}
\end{equation}
The special cases of the operator norm and the Frobenius norm give
\begin{eqnarray}
\underset{1 \leq j \leq n}{\max} |d_j (A) - d_j (B)| &\!\!\!\leq\!\!\!& \left ( 
\|A\|^{1/2} 
+ \|B\|^{1/2} \right ) \|A-B\|^{1/2}, \label{eq30}\\
\sqrt{2} \left (\sum_{j=1}^{n} |d_j(A) - d_j(B)|^2 \right )^{1/2} 
&\!\!\!\leq\!\!\!& \left 
(\|A\|^{1/2} + \|B\|^{1/2} \right ) \left (\tr \,|A-B|\right 
)^{1/2}.   \label{eq31}
\end{eqnarray}
(Here $|X|$ denotes the matrix absolute value defined as $|X| = (X^T X)^{1/2}.$)
\end{theorem}
The rest of the paper is organised as follows. In Section 2 we give a proof of 
Theorem \ref{thm7} and in Section 3 of Theorems \ref{thm1},\ref{thm3} and 
\ref{thm4}. In Section 4 we prove Theorem \ref{thm5}, and in 
Section 5 we prove Theorem \ref{thm6}. Some other results are proved along the 
way either as prerequisites or as supplements.

Let us recall here two facts about symplectic eigenvalues and associated pairs of eigenvectors. The imaginary numbers $\pm i d_j(A)$, $1\le j\le n$, constitute the set of eigenvalues of the skew-symmetric matrix $A^{1/2}JA^{1/2}$. To each $d_j(A)$ there corresponds a pair of vectors $u_j$, $v_j$ in $\mathbb{R}^{2n}$ such that 
\begin{equation*}
Au_j=d_j(A)Jv_j,\,\, Av_j=-d_j(A)Ju_j.
\end{equation*}
We may normalize these vectors so that the Euclidean inner product $\langle u_j,Jv_j\rangle=1$. Then we call $(u_j,v_j)$ a {\it symplectic eigenvector pair} corresponding to the symplectic eigenvalue $d_j$. Together, these $2n$ vectors constitute a {\it symplectic eigenbasis} for $\mathbb{R}^{2n}$; i.e., 
\begin{equation*}
\langle u_i, Ju_j\rangle=\langle v_i,Jv_j\rangle=0\textrm{ for all }i,j,
\end{equation*}
and
\begin{equation*}
\langle u_i,Jv_j\rangle=\delta_{ij}\textrm{ for all }i,j.
\end{equation*}

\section{Proof of Theorem \ref{thm7}}

A norm $\vertiii{\cdot}$ on $\mathbb{M}(n)$ is called unitarily invariant if 
$\vertiii{UXV} = \vertiii{X},$ for all $X \in \mathbb{M}(n)$ and for all 
unitary 
matrices $U,V.$ If $X,Y,Z$ are any three matrices, then $\vertiii{XYZ} \leq 
\|X\| \vertiii{Y} \|Z\|.$ See Chapter IV of \cite{rbh} for properties of 
such norms.

Let $A$ be a Hermitian matrix and $\Eig^{\downarrow}(A)$ the diagonal matrix 
whose diagonal entries are the decreasingly ordered eigenvalues of $A.$ By the 
famous Lidskii-Wielandt theorem (see (IV.62)) in \cite{rbh}) we have
$$\vertiii{\Eig^{\downarrow}(A) - \Eig^{\downarrow}(B)} \leq \vertiii{A-B}.$$
Now let $A \in \mathbb{P}(2n).$ The symplectic eigenvalues $d_j(A)$ with 
their negatives are the eigenvalues of the Hermitian matrix $i A^{1/2} 
JA^{1/2}.$ So, from the Lidskii-Wielandt theorem we obtain, for any $A,B$ in 
$\mathbb{P}(2n)$
\begin{eqnarray*}
 \vertiii{\widehat{D}(A)-\widehat{D}(B)} &\!\!\!\!\leq\!\!\!\!& 
\vertiii{A^{1/2} JA^{1/2} - B^{1/2} J B^{1/2}} \\
&\!\!\!\!\leq\!\!\!\!& \vertiii{A^{1/2} J A^{1/2} - A^{1/2} J B^{1/2}} + 
\vertiii{A^{1/2} J B^{1/2} - B^{1/2} J B^{1/2}} \\
&=& \vertiii{A^{1/2} J \left ( A^{1/2} - B^{1/2} \right )} + \vertiii{\left 
(A^{1/2} - B^{1/2} \right ) J B^{1/2}}  \\
&\leq& \| A^{1/2} J \| \,\,\vertiii{A^{1/2} - B^{1/2}} + \vertiii{A^{1/2} - 
B^{1/2}} \|JB^{1/2} \|\\
&=& \left ( \|A^{1/2}\| + \|B^{1/2}\| \right ) \vertiii{A^{1/2} - B^{1/2}}. 
\end{eqnarray*}
By theorem X.1.3 in \cite{rbh}
$$\vertiii{A^{1/2} - B^{1/2}} \leq \vertiii{|A-B|^{1/2}}$$
Combining these inequalities we obtain \eqref{eq29}. Using the definitions of 
$\|\cdot\|$ and $\|\cdot\|_2,$ we get \eqref{eq30} and \eqref{eq31} from this.

\vskip0.2in
\noindent{\bf Example} \,\, Let $\gamma$ be a positive number, and let
$$A = \left [\begin{array}{cc} \gamma I & O \\ O & I  \end{array} \right ], 
\quad B = \left [ \begin{array}{cc} I & O \\ O & I \end{array} \right ].$$
 Then
$$\widehat{D}(A) = \left [ \begin{array}{cc} \gamma^{1/2}I & O \\ O & 
\gamma^{1/2} I  \end{array} \right ], \quad  \widehat{D}(B) = \left [ 
\begin{array}{cc} I & O \\ O & I  \end{array} \right ]$$
So, if $\gamma \ge 1,$ then $\|\widehat{D}(A) - \widehat{D} (B) \| = 
\gamma^{1/2}-1,$ and $\|A-B\| = \gamma - 1.$ This shows that for large 
$\gamma,$ both 
$\|\widehat{D}(A)-\widehat{D}(B)\|$ and $\|A-B\|^{1/2}$ are close to 
$\gamma^{1/2}.$ Thus the bound given by Theorem \ref{thm7} has the right order.

\section{Proofs of Theorems \ref{thm1}, \ref{thm3} and \ref{thm4}}
We  first prove the relation \eqref{eq9} in the special case $t=2.$ We use 
that to establish \eqref{eq15}, and then derive \eqref{eq8} from it. From this 
we obtain \eqref{eq9} for all $t \ge 1.$ Finally, we use the relation 
\eqref{eq15} to get 
 Theorem \ref{thm4}.

We use two elementary properties of the operator norm $\|\cdot\|.$ For any 
matrix $X$ we have $\|X\|^2 = \|XX^T\| = \|X^TX\|.$ If $X$ and $Y$ are any two 
matrices such that $XY$ is normal, then $\|XY\| \leq \|YX\|.$ This is so 
because the norm of a normal matrix is equal to its spectral radius, the 
spectral radius of $XY$ and $YX$ are equal, and in general the norm of $X$ is 
bigger than the spectral radius of $X.$

Now let $A \in \mathbb{P}(2n).$ Then $\widehat{d_1} (A)$ is the maximum 
eigenvalue of $i A^{1/2} J A^{1/2}.$ So, using the properties stated above, we 
get
\begin{eqnarray*}
\widehat{d_1}^2 (A) &=& \| A^{1/2} J A^{1/2} \|^2 = \|A^{1/2} J A J^T A^{1/2} 
\| \\
&\leq& \|AJAJ^T\| = \|AJA\| =  \widehat{d_1} (A^2).
\end{eqnarray*}
Apply these same considerations to the $k$th antisymmetric tensor power 
$\Lambda^k A.$ This gives
\begin{eqnarray*}
\prod_{j=1}^{k} \widehat{d_j}^2 (A) &=& \| \Lambda^k \left ( A^{1/2} J A^{1/2} 
\right ) \|^2 \\
&\leq& \| \Lambda^k A \,\,\Lambda^k J \,\,\Lambda^k A \,\,\Lambda^k J^T \|\\
&=& \|\Lambda^k (AJA) \| = \prod_{j=1}^{k} \widehat{d_j} (A^2),
\end{eqnarray*}
for all $1 \leq k \leq 2n.$ When $k=2n,$ the two extreme sides of the last 
inequality are equal to $\det A^2.$ This establishes \eqref{eq9} for the 
special case $t=2.$

Now let $A,B$ be any two elements of $\mathbb{P}(2n)$ and put
$$U = \left ( A^{-1/2} B A^{-1/2} \right )^{1/2} A^{1/2} B^{-1/2}. $$
Then $U^TU=I,$ and so $U$ is orthogonal. From the formula \eqref{eq12} we see 
that
\begin{equation}
 G(A,B) = A^{1/2} U B^{1/2} = B^{1/2} U^T A^{1/2}, \label{eq32}
\end{equation}
where the last equality follows from the fact that $G(A,B) = G (A,B)^T.$

\noindent For brevity put $G = G (A,B).$ By what we have already proved
$$ \widehat{d_1}^2 (G) \leq \widehat{d_1} (G^2) = \|G JG \|.$$
Using \eqref{eq32} we see that
\begin{eqnarray*}
\| GJG \| &=& \|A^{1/2} U B^{1/2} J B^{1/2} U^T A^{1/2} \| \\
&\leq&  \|A U B^{1/2} J B^{1/2} U^T \|  \\
&\leq&\|AU\| \,\,\|B^{1/2} J B^{1/2} U^T\|   \\
&=& \|A\| \,\,\|B^{1/2} J B^{1/2} \| = \|A\| \widehat{d_1} (B).
\end{eqnarray*}
Thus, we have
\begin{equation}
\widehat{d_1}^2 (G) \leq \|A\| \widehat{d_1} (B).    \label{33}
\end{equation}
By the invariance of $G(A,B)$ under congruence transformations, we have for 
every $M \in GL (2n)$
$$M^T G(A,B) M = G (M^T AM, M^T BM). $$
If $M$ is symplectic, then the symplectic eigenvalues of $M^T G(A,B)M$ are the 
same as those of $G(A,B).$ So
\begin{equation}
\widehat{d_1} (G(A,B)) = \widehat{d_1} (G(M^T AM, M^T BM)).   \label{eq34}
\end{equation}
Choose $M \in Sp (2n)$ so that $M^T AM = \left [\begin{array}{cc} D & O \\ O & 
D \end{array} \right ].$ Then $\|M^T AM\| = \|D\| = \widehat{d_1} (A).$ Using 
this fact we obtain from \eqref{33} and \eqref{eq34}
\begin{equation}
\widehat{d_1} (G(A,B)) \leq (\widehat{d_1} (A)\widehat{d_1} (B) )^{1/2}.    
\label{eq35}
\end{equation}
Once again, applying this to $\Lambda^kA$ and $\Lambda^kB$  we obtain the 
$\log$ majorisation \eqref{eq16}.

The equation \eqref{eq14} gives a natural parametrisation of the geodesic 
joining $A$ and $B.$ Hence
$$A \#_{1/4} B = A \#_{1/2} (A \#_{1/2} B). $$
So, from \eqref{eq35} we obtain 
\begin{eqnarray*}
\widehat{d_1} (A \#_{1/4} B) &\leq& \widehat{d_1}^{1/2} 
(A)\,\,\widehat{d_1}^{1/2} (A \#_{1/2} B) \\
&\leq& \widehat{d_1}^{1/2} (A) \,\,\widehat{d_1}^{1/4} (A) 
\,\,\widehat{d_1}^{1/4} (B) \\
&=& \widehat{d_1}^{3/4} (A) \,\,\widehat{d_1}^{1/4} (B).
\end{eqnarray*}
This argument can be repeated to show that
\begin{equation}
\widehat{d_1} (A \#_t B) \leq  \widehat{d_1}^{1-t} (A) \,\,  \widehat{d_1}^t 
(B),   \label{eq36}
\end{equation}
for all dyadic rationals $t$ in $[0,1].$ By the continuity of symplectic 
eigenvalues, this is then true for all $t$ in $[0,1].$ Using antisymmetric 
tensor powers, we obtain \eqref{eq15} from \eqref{eq36}. This completes the 
proof of Theorem \ref{thm3}.

The inequality \eqref{eq8} is a special case of \eqref{eq15}, since $I \#_t A = 
A^t$ for all $0 \leq t \leq 1.$ If $t \ge 1,$ let $s = 1/t.$ Then from 
\eqref{eq8} we have $\widehat{d} (A^s) \prec_{\log} \widehat{d\,\,}^s (A).$ 
Replace 
$A^s$ by $A$ to obtain \eqref{eq9}. This completes the proof of Theorem 
\ref{thm1}.

Now we turn to Theorem \ref{thm4}. It was shown by E. Cartan that the 
minimising problem in \eqref{eq18} has a unique solution, and this is also the 
unique positive definite solution of the equation
\begin{equation}
 \sum_{j=1}^{m} \log \left ( X^{1/2} A_j^{-1} X^{1/2} \right ) = 0. \label{eq37}
\end{equation}
See e.g. \cite{rbh1}, \cite{n}. A direct description of $G$ suitable for some 
operator theoretic problems has been found recently. This describes 
$G(A_1,\ldots,A_m)$ as the limit of a ``walk'' in the Riemannian metric space 
$\mathbb{P}.$ Consider the sequence $S_k$ defined as
\begin{eqnarray*}
 S_1 &=& A_1 \\
S_2 &=& S_1 \#_{1/2} A_2 \\
\vdots && \\
S_{k+1} &=& S_k \#_{1/k+1} A_{\overline{k+1}}, \quad 
\mbox{where}\,\,\overline{k}=k \,(\mbox{mod}\,m).
\end{eqnarray*}
Then it turns out that
\begin{equation}
G(A_1, \ldots, A_m) = \underset{k \rightarrow \infty}{\lim} \,S_k.    
\label{eq38}
\end{equation}
A stochastic version of this was proved in \cite{ll} and some 
simplifications made in \cite{bk}. The statement \eqref{eq38} was first proved in \cite{hol} and then a considerably simpler proof given in \cite{lp}. The 
effectiveness of this formula stems from the fact that it gives $G$ as a limit 
of the binary mean operation $\#$ rather than the solution to an $m$-variable 
minimisation problem as in \eqref{eq17}, or as the solution of an $m$-variable 
nonlinear matrix equation as in \eqref{eq37}.

The majorisation relation \eqref{eq20} can be derived now from \eqref{eq15}. 
First use it to get a majorisation for $\widehat{d} (S_k)$ as  in the 
proof of \eqref{eq15}, and then take the limit as $k \rightarrow \infty.$ The 
proof of the weighted version \eqref{eq19} is a modification of this idea. We 
can proceed either as in \cite{bk}, first proving it for rational weights and 
then taking a limit, or as in \cite{lp} where the definition of $S_k$ is 
modified to include weights.

An element $A$ of $\mathbb{P}(2n)$ is called a {\it Gaussian matrix} (or, more precisely, the covariance matrix corresponding to a Gaussian state) if $A \pm 
\frac{i}{2} J$ is positive definite. Using \eqref{eq1} one can see that this 
condition  is equivalent to saying that $d_1 (A) \ge 1/2.$ Gaussian matrices 
are being intensely studied in the current literature on quantum information. 
Theorems \ref{thm1},  \ref{cor2},  \ref{thm4} have an interesting corollary.

\begin{corollary}\label{cor8}
\begin{itemize} 
\item[]
\item[(i)] Let $A$ be a Gaussian matrix. Then for every $0 \leq 
t \leq 1,$ $A^t$ is Gaussian. 
\item[(ii)] Let $A,B$ be Gaussian matrices. Then every point on the Riemannian 
geodesic $A \#_t B,$ $0 \leq t \leq 1$ is a Gaussian matrix. Thus the set of 
Gaussian matrices is a geodesically convex set in the Riemannian metric space 
$(\mathbb{P}(2n), \delta).$
\item[(iii)] The geometric mean of any $m$-tuple of Gaussian matrices is 
Gaussian.
\end{itemize}
\end{corollary}

\begin{proof}
Imbedded in \eqref{eq10} is the inequality $d_1(A^t) \ge d_1^t (A)$ for $0 \leq 
t \leq 1.$ So, the statement (i) follows. Likewise (ii) and (iii) follow from 
Theorems \ref{thm3} and \ref{thm4}.
\end{proof}

\section{Variational principles and a proof of \\Theorem \ref{thm5}}

The Courant-Fischer-Weyl minmax principle is one of the most powerful tools in 
the analysis of eigenvalues of Hermitian matrices. Such a principle is known 
also for symplectic eigenvalues. We state it in a form suitable for us and, for 
the convenience of the reader, indicate its proof. The idea is borrowed from  \cite{hz},p.39.

We denote the usual Euclidean inner product on $\mathbb{R}^m$ or on 
$\mathbb{C}^m$ by $\langle \cdot , \cdot \rangle.$ In the latter case we assume 
that the inner product is conjugate linear in the first variable. Given $A \in 
\mathbb{P}(2n),$ introduce another inner product on $\mathbb{C}^{2n}$ by 
putting 
\begin{equation}
(x,y) = \langle x, Ay \rangle.   \label{eq39}
\end{equation}
Call the resulting inner product space $\mathcal{H}.$ Let $A^{\#} = i A^{-1} 
J.$ Then
$$(x, A^{\#} y ) = i \langle x, Jy \rangle = \left ( A^{\#} x, y \right ). $$
So $A^{\#}$ is a  Hermitian operator on $\mathcal{H}.$ The symplectic 
eigenvalues of $A^{-1}$ arranged in decreasing order are $\frac{1}{d_1(A)} \ge 
\frac{1}{d_2(A)} \ge \cdots \ge \frac{1}{d_n (A)}.$ The (usual) eigenvalues of 
$A^{\#}$ are 
$$\frac{1}{d_1(A)}  \ge \frac{1}{d_2(A)} \ge \cdots \ge \frac{1}{d_n (A)} \ge 
\frac{-1}{d_n (A)} \ge \cdots \ge \frac{-1}{d_1(A)}   $$
So, from the usual minmax principle  (Corollary III.1.2 in \cite{rbh}) applied to 
$A^{\#}$ we get the following.
\vskip0.1in
\noindent{\bf The minmax principle for symplectic eigenvalues.}
\vskip0.1in
 Let $A \in 
\mathbb{P}(2n).$ Then for $1 \leq j \leq n$
\begin{equation}
 \frac{1}{d_j (A)} = \underset{{\mathcal{M} \subset 
\mathbb{C}^{2n}}\atop{\dim \mathcal{M}=j}}{\max} \,\,\,\underset{{x \in 
\mathcal{M}} \atop{\langle x, Ax \rangle = 1}}{\min} \,\,\, \langle x, i J x 
\rangle, \label{eq40}
\end{equation}
and also
\begin{equation}
 \frac{1}{d_j (A)} = \underset{{\mathcal{M} \subset 
\mathbb{C}^{2n}}\atop{\dim \mathcal{M}=2n-j+1}}{\min} \,\,\,\underset{{x \in 
\mathcal{M}} \atop{\langle x, Ax \rangle = 1}}{\max} \,\,\, \langle x, i J x 
\rangle, \label{eq41}
\end{equation}

One of the important corollaries of the minmax principle for Hermitian matrices 
is the interlacing principle for eigenvalues of $A$ and those of a principal 
submatrix. So it is for symplectic eigenvalues: 
\vskip0.1in
\noindent{\bf The interlacing theorem for symplectic eigenvalues.} 
\vskip0.1in
Let $A\in\mathbb{P}(2n)$. Partition $A$ as $A=\begin{bmatrix}A_{ij}\end{bmatrix}$ where each $A_{ij},$ $i,j=1,2$, is an $n\times n$ matrix. A matrix $B\in\mathbb{P}(2n-2)$ is called an {\it s-principal submatrix} of $A$ if $B=\begin{bmatrix}B_{ij}\end{bmatrix}$, and each $B_{ij}$ is an $(n-1)\times (n-1)$ principal submatrix of $A_{ij}$ occupying the same position in $A_{ij}$ for $i,j=1,2$. In other words, $B$ is obtained from $A$ by deleting, for some $1\le i\le n$, the $i$th and $(n+i)$th rows and columns of $A$. Then 
\begin{equation}
d_j (A) \leq d_j(B) \leq d_{j+2}(A), \quad 1 \leq j \leq n-1,  \label{eq42}
\end{equation}
where we adopt the convention that $d_{n+1} (A) = \infty.$

The proof is similar to the one in the classical Hermitian case. See 
\cite{rbh},p.59. This observation has been made in \cite{ktv}.

Now we come to the proof of Theorem \ref{thm5}. We begin with a proof of the 
inequalities \eqref{eq26}.  From the condition 
$AG^T-BC^T=I$ in \eqref{eq24}, we 
have for $1 \leq i \leq n$
\begin{eqnarray*}
1 &=& \sum_{j=1}^{n} (a_{ij} g_{ij} - b_{ij} c_{ij}) \\
&\leq& \sum_{j=1}^{n} \frac{1}{2} \left (a_{ij}^2 + g_{ij}^2 \right ) + 
\sum_{j=1}^{n} \frac{1}{2} \left (b_{ij}^2 + c_{ij}^2 \right )\\
&=& \sum_{j=1}^{n} \widetilde{m}_{ij}.
\end{eqnarray*}
Applying the same argument to $M^T$ we see that the second inequality in 
\eqref{eq26} also holds. Now we can prove Part (i) of Theorem \ref{thm5} in the 
special case $k=n.$ Without loss of generality, we may assume that $A = 
\widehat{D}= \left [\begin{array}{cc}  D & O \\ O & D \end{array} \right ].$ 
Let $M$ be any element of $Sp(2n)$ and decompose it as $M = \left 
[\begin{array}{cc}  P & Q \\ R & S \end{array} \right ]$ 
according to the rules \eqref{eq23} and \eqref{eq24}. Then 
\begin{eqnarray*}
\tr \,\, M^T \widehat{D} M &=& \tr \,\,\left (P^T D P + Q^T DQ + R^T DR+S^T DS 
 \right ) \\
&=& \sum_{i=1}^{n} \,\, d_i(A) \,\sum_{j=1}^{n}  \left (p_{ij}^2 + q_{ij}^2 + 
r_{ij}^2 + s_{ij}^2 \right ) \\
&=& \sum_{i=1}^n \, d_i (A) \sum_{j=1}^{n} \left ( 2 \widetilde{m}_{ij} \right 
) \\
&\ge& 2 \sum_{i=1}^n \, \, d_i (A),
\end{eqnarray*}
using \eqref{eq26}. When $M=I,$ the two extreme sides of this equality are 
equal. Thus
\vskip0.1in
\begin{equation}
\underset{M \in Sp (2n)}{\min} \,\tr \,M^T AM = 2 \sum\limits_{j=1}^{n} 
\,d_j (A). \label{eq43}
\end{equation}
This is the special case of \eqref{eq21} when $k=n.$

Let $M$ be a $2n\times 2k$ matrix satisfying the condition $M^TJ_{2n}M=J_{2k}$. Partition $M$ as $M=\begin{bmatrix}P^\prime & Q^\prime\\
R^\prime & S^\prime\end{bmatrix},$ where each block is an $n\times k$ matrix. Then we can find a $2n\times 2n$ symplectic matrix $L=\begin{bmatrix}P & Q\\
R & S\end{bmatrix}$ in which each block is an $n\times n$ matrix and the first $k$ columns of $P,$ $Q$, $R,$ $S$ are the columns of $P^\prime$, $Q^\prime$, $R^\prime$, $S^\prime$, respectively. The matrix $M^TAM$ is then a $2k\times 2k$ s-principal submatrix of $L^TAL$.

The symplectic eigenvalues of $L^T AL$ are $d_1 (A) \leq d_2 (A) \leq \cdots 
\leq d_n (A).$ Let those of $M^T AM$ be $d_1^{\prime} \leq d_2^{\prime}  \leq 
\cdots \leq d_k^{\prime}.$ By the interlacing principle $d_j^{\prime} \ge 
d_j(A)$ for $1 \leq j \leq k.$

Now we can complete the proof of Theorem \ref{thm5}. First from the special 
case of (i) proved above we can see that
$$\tr \, M^T AM \ge 2 \sum\limits_{j=1}^{k} \,\,d_j^{\prime}.$$
Then from the interlacing principle we see that
\begin{equation}
\tr \,\,M^T AM \ge 2 \sum\limits_{j=1}^{k} \,\,d_j (A). \label{eq44}
\end{equation}
By the same arguments we see that
\begin{equation}
\det \,M^T AM \ge \prod_{j=1}^{k} \,d_j^{\prime \,2} \ge \prod_{j=1}^{k} 
\,d_j^2 (A). \label{eq45}
\end{equation}
There is equality in the inequalities \eqref{eq44} and \eqref{eq45} when 
$M$ is the matrix 
whose columns are the symplectic eigenvectors of $A$ corresponding to $d_1 
(A) , \ldots, d_k(A).$ This proves Theorem \ref{thm5}. $\blacksquare$

An immediate corollary of this theorem is that if $A, B \in \mathbb{P}(2n),$ 
then for all  $1 \leq k \leq n,$ we have
\begin{eqnarray}
\sum_{j=1}^{k} \,\,d_j (A+B) &\ge& \sum_{j=1}^{k} \,d_j (A) + \sum_{j=1}^{k} \, 
d_j(B), \label{eq46} \\
\prod_{j=1}^{k} \,\,d_j^2 (A+B) &\ge& \prod_{j=1}^{k} \,d_j^2 (A) + 
\prod_{j=1}^{k} \,d_j^2(B).\label{eq47}. 
\end{eqnarray}

\section{Proof of Theorem \ref{thm6}}

We use a theorem of Elsner and Friedland \cite{ef}. This says that 
if $R$ is an $n \times n$ matrix with singular values $s_1 (R) \ge \cdots \ge 
s_n (R),$ then there exist doubly stochastic matrices $P$ and $Q$ for which
\begin{equation}
s_n (R)^2 \,\,p_{ij} \leq |r_{ij}|^2 \leq s_1 (R)^2 \,\,q_{ij}    \label{eq48}
\end{equation}
for all $1 \leq i,$ $j \leq n.$

The {\it Euler decomposition theorem} says that every symplectic matrix $M$ can 
be decomposed as
\begin{equation}
 M = O_1 \left [ \begin{array}{cc} \Gamma & O \\ O & \Gamma^{-1} \end{array} 
\right ] \,O_2^T,  \label{eq49}
\end{equation}
where $O_1$ and $O_2$ are orthogonal and symplectic, and $\Gamma = \diag 
(\gamma_1, \ldots, \gamma_n)$ with
\begin{equation}
\gamma_1 \ge \gamma_2 \ge \cdots \ge \gamma_n \ge 1.   \label{eq50}
\end{equation}
There is a correspondence between $2n \times 2n$ real orthogonal symplectic 
matrices and $n \times n$ complex unitary matrices that tells us that we can 
find $n \times n$ unitary matrices $U$ and $V$ such that
\begin{equation}
U = X + i Y, \quad V = Z + i W,   \label{eq51}
\end{equation}
where $X, Y, Z, W$ are real, and 
\begin{equation}
O_1 = \left [ \begin{array}{cc} X  & -Y \\ Y & X \end{array} 
\right ], \quad  O_2 = \left [ \begin{array}{cc} Z  & -W \\ W & Z \end{array} 
\right ].   \label{eq52}
\end{equation}
Both the theorems cited above may be found in \cite{dms} or \cite{g}.

Using \eqref{eq23}, \eqref{eq49} and \eqref{eq52} we see that
\begin{eqnarray}
 A &=& X \Gamma Z^T + Y \Gamma^{-1} W^T, \quad B = X \Gamma W^T - Y \Gamma^{-1} 
Z^T, \nonumber \\
C & =& Y \Gamma Z^T - X \Gamma^{-1} W^T, \quad G = Y \Gamma W^T + X \Gamma^{-1} 
Z^T. \label{eq53}
\end{eqnarray}
From \eqref{eq51} we have
\begin{eqnarray}
 X &=& \frac{1}{2} \left ( U + \overline{U} \right ), \quad Y = \frac{1}{2i} 
\left ( U - \overline{U} \right ),  \nonumber\\
Z &=& \frac{1}{2} \left ( V + \overline{V} \right ), \quad W = \frac{1}{2i} 
\left ( V - \overline{V} \right ). \label{eq54}
\end{eqnarray}
Here $\overline{U}$ stands for the entrywise complex conjugate of $U.$ We will 
use the notation $U^{\ast}$ for $\overline{U}^T.$ Let
\begin{equation}
\Sigma = \frac{1}{2} (\Gamma + \Gamma^{-1} ), \quad \Delta = \frac{1}{2} 
(\Gamma - \Gamma^{-1})     \label{eq55}                                         
\end{equation}
Both are positive diagonal matrices. Let 
\begin{equation}
 \Sigma = \diag (\sigma_1, \ldots, \sigma_n), \quad \Delta = \diag 
(\delta_1, \ldots, \delta_n). \label{eq56}
\end{equation}
From the first equation in \eqref{eq53}, and the equations \eqref{eq54} and 
\eqref{eq55} we see after a little calculation that
\begin{equation}
A = \frac{1}{2} \left (U \Delta V^T + U \Sigma V^{\ast} + \overline{U} 
\Sigma V^T + \overline{U} \Delta V^{\ast} \right ). \label{eq57}
\end{equation}
Another calculation involving the entries of the matrices in \eqref{eq57} shows 
that
\begin{equation}
 a_{ij} = \sum_{k=1}^{n} \delta_k \Re (u_{ik} v_{jk}) + \sum_{k=1}^{n} \sigma_k 
\Re (u_{ik} \overline{v}_{jk}). \label{eq58}
\end{equation}
Similar calculations with the other three equations in \eqref{eq53} show that

\begin{eqnarray}
b_{ij} &=& \sum_{k=1}^{n} \delta_k \,\,\Im \, (u_{ik} v_{jk}) + \sum_{k=1}^{n} 
\sigma_k \,\Im (u_{ik} \overline{v}_{jk}), \label{eq59} \\
c_{ij} &=& \sum_{k=1}^{n} \delta_k \,\,\Im \, (u_{ik} v_{jk}) - \sum_{k=1}^{n} 
\sigma_k \,\Im (u_{ik} \overline{v}_{jk}), \label{eq60} \\
g_{ij} &=& -\sum_{k=1}^{n} \delta_k \,\,\Re \, (u_{ik} v_{jk}) + \sum_{k=1}^{n} 
\sigma_k \,\Re (u_{ik} \overline{v}_{jk}). \label{eq61} 
\end{eqnarray}
Squaring the equations \eqref{eq58}-\eqref{eq61}, adding them and simplifying 
the resulting expression, we see that
\begin{equation}
\frac{1}{2} \left (a_{ij}^2 + b_{ij}^2 + c_{ij}^2 + g_{ij}^2 \right 
) = \left |\sum_{k=1}^{n} \delta_k u_{ik} v_{jk} \right |^2 +  
\left |\sum_{k=1}^{n} \sigma_k u_{ik} \overline{v}_{jk} \right |^2. \label{eq62}
\end{equation}
This shows that
\begin{equation}
\widetilde{m}_{ij} \ge \left |\sum_{k=1}^{n} \sigma_k u_{ik} \overline{v}_{jk} 
\right |^2. \label{eq63}
\end{equation}
Now let $R=U \Sigma V^{\ast}.$ Then the right-hand side of \eqref{eq63} is 
equal to $\left |r_{ij}\right |^2.$ From \eqref{eq55} and \eqref{eq56} we see 
that 
the smallest singular value of $R$ is $\sigma_n = \frac{1}{2} \left ( \gamma_n 
+ \gamma_n^{-1} \right ).$ So, from \eqref{eq48} we see that there exists a 
doubly stochastic matrix $P$ such that
\begin{equation}
 \widetilde{m}_{ij} \ge \left | r_{ij} \right |^2 \ge \sigma_n^2 p_{ij}. 
\label{eq64}
\end{equation}
Since $\frac{1}{2} (x+x^{-1}) \ge 1$ for any positive number $x,$ we have 
$\sigma_n \ge 1.$ So, it follows from \eqref{eq64} that $\widetilde{M}$ is 
doubly superstochastic. This proves the first statement of Theorem \ref{thm6}.

Now suppose $M$ is symplectic and orthogonal. We have noted earlier that then 
there exists a complex unitary matrix $U = X + i Y$ such that
$$M = \left [\begin{array}{cc} X & - Y \\ Y & X \end{array} \right ]. $$
It is clear from this that the matrix $\widetilde{M}$ associated with this via 
\eqref{eq20} is doubly stochastic.

To prove the converse, return to the relation \eqref{eq62}. We have already 
seen that if the second term on the right-hand side is equal to $\left | r_{ij} 
\right |^2,$ then the matrix $R$ dominates entrywise a doubly stochastic matrix 
$P.$ So, a necessary condition for $\widetilde{M}$ to be doubly stochastic is 
that
$$\left |\sum_{k=1}^{n} \delta_k u_{ik} v_{jk}  \right |^2 = 0 \quad \mbox{for 
all}\quad i,j.$$
Translated to matrices, this says that $U \Delta V^T = 0.$ By the definition of 
$\Delta$ in \eqref{eq56}, this is equivalent to the condition $\gamma_j - 
\gamma_j^{-1} = 0$ for $1 \leq j \leq n;$ or in other words $\gamma_j = 1$ for 
$1 \leq j \leq n.$ In turn, this means that $M$ is orthogonal. The proof of 
Theorem \ref{thm6} is complete. \hfill{$\blacksquare$}

For the theory of majorisation and the role of doubly superstochastic matrices 
in it we refer the reader to the comprehensive treatise \cite{moa}. 

Let $x = (x_1, \ldots, x_n)$ be any element of $\mathbb{R}^n$ and let 
$x^{\uparrow} = (x_1^{\uparrow}, \ldots, x_n^{\uparrow}$) be the vector 
obtained from $x$ by rearranging its coordinates in increasing order
$$x_1^{\uparrow} \leq x_2^{\uparrow} \leq \cdots \leq  x_n^{\uparrow}.   $$ 
We say $x$ is {\it supermajorised} by $y,$ in symbols $x \prec^w y,$ if for $1 
\leq k \leq n$ 
\begin{equation}
 \sum_{j=1}^{k} \, x_j^{\uparrow} \ge \sum_{j=1}^{k} \,y_j^{\uparrow}. 
\label{eq65}
\end{equation}
A fundamental theorem in the theory of  majorisation says that the following 
two conditions are equivalent:
\begin{itemize}
\item[(i)] An $n \times n$ matrix $A$ is doubly superstochastic.
\item[(ii)] $Ax \prec^{w} x$ for every positive $n$-vector $x.$
\end{itemize}
Inequalities like \eqref{eq46} express a supermajorisation. An alternative 
proof of Theorem \ref{thm5}(i) can be obtained using Theorem \ref{thm6}.

\section{Some remarks}
Let $m_1,m_2,\ldots,m_k$ be positive integers, and let $n=m_1+m_2+\cdots +m_k$. If $A_j,$ $1\le j\le k$, are $m_j\times m_j$ matrices, we write $\oplus A_j$ for their usual direct sum. This is the $n\times n$ block-diagonal matrix with entries $A_1,\ldots, A_k$ on its diagonal and zeros elsewhere. Given an $n\times n$ matrix $A$ partitioned into blocks as $A=\begin{bmatrix}A_{ij}\end{bmatrix}$, where the diagonal blocks $A_{jj}$ are $m_j\times m_j$ in size, the {\it pinching} of $A$ is the block diagonal matrix $\oplus A_{jj}$. This is denoted by $\mathcal{C}(A)$. We introduce a version of direct sum and pinching adapted to the symplectic setting. Let 
\begin{equation*}
A_j=\begin{bmatrix}P_j & Q_j\\
R_j & S_j\end{bmatrix},\ 1\le j\le k
\end{equation*}
be $2m_j\times 2m_j$ matrices partitioned into blocks of size $m_j\times m_j$. The {\it $s$-direct sum} of $A_j$ is defined to be the $2n\times 2n$ matrix
\begin{equation*}
\oplus^s A_j=\begin{bmatrix}\oplus P_j & \oplus Q_j\\
\oplus R_j & \oplus S_j\end{bmatrix}.
\end{equation*}
Then, one can see that $\oplus^s J_{2m_j}=J_{2m}$, the $s$-direct sum of symplectic matrices is symplectic, and the $s$-direct sum of positive definite matrices is positive definite. If $A$ is a $2n\times 2n$ and $B$ a $2m\times 2m$ positive definite matrix, then the symplectic eigenvalues of their $s$-direct sum are the symplectic eigenvalues of $A$ and $B$ put together. Let $\mathcal{C}$ be a pinching on $n\times n$ matrices. Then we define the {\it $s$-pinching} of a $2n\times 2n$ matrix $A=\begin{bmatrix}P & Q\\
R & S\end{bmatrix}$ as 
\begin{equation*}
\mathcal{C}^s (A)=\begin{bmatrix}\mathcal{C}(P) & \mathcal{C}(Q)\\
\mathcal{C}(R) & \mathcal{C}(S)\end{bmatrix}.
\end{equation*}
If $A$ is positive definite, then so is $\mathcal{C}^s(A)$. Our next theorem gives a majorisation relation between the symplectic eigenvalues of $A$ and those of $\mathcal{C}^s(A)$. 
\begin{theorem}\label{thma}
Let $A$ be any element of $\mathbb{P}(2n)$ and let $\mathcal{C}^s(A)$ be an $s$-pinching of $A$. Then 
\begin{equation}
\hat{d}(\mathcal{C}^s(A))\prec^w \hat{d}(A).\label{eqa}
\end{equation}
\end{theorem}

\begin{proof}
It is enough to consider the case when $n=m_1+m_2$ and $\mathcal{C}$ is a pinching into two blocks; i.e., 
\begin{equation*}
T=\begin{bmatrix}T_{11} & T_{12}\\
T_{21} & T_{22}\end{bmatrix}\textrm{ and }\mathcal{C}(T)=\begin{bmatrix}T_{11} & O\\
O & T_{22}\end{bmatrix}.
\end{equation*}
The general case can be derived by repeated applications of such pinchings. Partition the $2n\times 2n$ positive definite matrix $A$ as 
\begin{equation*}
A=\begin{bmatrix}P & Q\\
Q^T & R\end{bmatrix}=\begin{bmatrix}P_{11} & P_{12} & & Q_{11} & Q_{12}\\
P_{21} & P_{22} & & Q_{21} & Q_{22}\\
 & & & & \\
Q_{11}^T & Q_{21}^T & & R_{11} & R_{12}\\
Q_{12}^T & Q_{22}^T & & R_{21} & R_{22}\end{bmatrix},
\end{equation*}
where $P_{11}$ and $R_{11}$ are $m_1\times m_1$, and $P_{22}$ and $R_{22}$ are $m_2\times m_2$ matrices with $m_1+m_2=n$. Then 
\begin{equation*}
\mathcal{C}^s(A)=\begin{bmatrix}P_{11} & O & & Q_{11} & O\\
O & P_{22} & & O & Q_{22}\\
 & & & & \\
Q_{11}^T & O & & R_{11} & O\\
O & Q_{22}^T & & O & R_{22}\end{bmatrix}.
\end{equation*}
Evidently, $\mathcal{C}^s(A)$ is the $s$-direct sum of a $2m_1\times 2m_1$ matrix $B$ and a $2m_2\times 2m_2$ matrix $C$ defined as
\begin{equation*}
B=\begin{bmatrix}P_{11} & Q_{11}\\
Q_{11}^T & R_{11}\end{bmatrix}, \ C=\begin{bmatrix}P_{22} & Q_{22}\\
Q_{22}^T & R_{22}\end{bmatrix}.
\end{equation*}
The symplectic eigenvalues of $\mathcal{C}^s(A)$ are the symplectic eigenvalues of $B$ and those of $C$ put together. So, given $1\le k\le n$, there exist $k_1,k_2$ such that $1\le k_1\le m_1$, $1\le k_2\le m_2$, $k_1+k_2=k$ and 
\begin{equation}
\sum\limits_{j=1}^{k}d_j(\mathcal{C}^s(A))=\sum\limits_{j=1}^{k_1}d_j(B)+\sum\limits_{j=1}^{k_2}d_j(C).\label{eqa1}
\end{equation}
Using \eqref{eq21} we can choose a $2m_1\times 2k_1$ matrix $M_1$ and a $2m_2\times 2k_2$ matrix $M_2$ such that
\begin{equation*}
M_1^TJ_{2m_1}M_1=J_{2k_1},\ M_2^TJ_{2m_2}M_2=J_{2k_2},
\end{equation*}
and
\begin{equation}
2\sum\limits_{j=1}^{k_1}d_j(B)=\tr\, M_1^TBM_1,\qquad 2\sum\limits_{j=1}^{k_2}d_j(C)=\tr\, M_2^TCM_2.\label{eqa2}
\end{equation}
Let
\begin{equation*}
M_1=\begin{bmatrix}P_1 & Q_1\\
R_1 & S_1\end{bmatrix}, \ M_2=\begin{bmatrix}P_2 & Q_2\\
R_2 & S_2\end{bmatrix},
\end{equation*}
where $P_1,Q_1,R_1,S_1$ are $m_1\times k_1$ matrices and $P_2,Q_2,R_2,S_2$ are $m_2\times k_2$ matrices, and then let 
\begin{equation*}
M=\begin{bmatrix}P_1 & O & Q_1 & O\\
O & P_2 & O & Q_2\\
R_1 & O & S_1  & O\\
O & R_2 & O & S_2\end{bmatrix}.
\end{equation*}
 Using the relations \eqref{eq24} it can be seen that the $2n\times 2k$ matrix $M$ satisfies the equation
\begin{equation*}
M^TJ_{2n}M=J_{2k}.
\end{equation*}
Further,
\begin{equation}
\tr\, M_1^TBM_1+\tr\, M_2^TCM_2=\tr\, M^TAM.\label{eqa3}
\end{equation}
Combining \eqref{eqa1}, \eqref{eqa2} and \eqref{eqa3} we see that 
\begin{equation*}
2\sum\limits_{j=1}^{k}d_j(\mathcal{C}^s(A))=\tr\, M^TAM.
\end{equation*}
It follows from \eqref{eq21} that 
\begin{equation*}
2\sum\limits_{j=1}^{k}d_j(\mathcal{C}^s(A))\ge 2\sum\limits_{j=1}^{k}d_j(A).
\end{equation*}
This proves \eqref{eqa}.
\end{proof}

Using standard arguments from the theory of majorisation one has the following consequence.

\begin{corollary}\label{corc}
Let $f:\mathbb{R}^n_+\to \mathbb{R}$ be any function that is permutation invariant, concave and monotone increasing. Then
\begin{equation}
f\bigl(d_1(\mathcal{C}(A)), \ldots, d_n(\mathcal{C}(A))\bigr)\ge f\bigl(d_1(A),\ldots,d_n(A)\bigr).\label{eqc}
\end{equation}
\end{corollary}

Among functions that satisfy the requirements of Corollary \ref{corc} are
\begin{equation*}
f(x_1,\ldots,x_n)=s_k(x_1,\ldots,x_n),
\end{equation*}
and
\begin{equation*}
f(x_1,\ldots,x_n)=\bigl(s_k(x_1,\ldots,x_n)\bigr)^{1/k},
\end{equation*}
where $s_k$ are the elementary symmetric polynomials, $1\le k\le n$. The functions 
\begin{equation*}
f(x_1,\ldots,x_n)=\sum\limits_{j=1}^{n}\frac{x_j}{1+x_j}, 
\end{equation*}
\begin{equation*}
f(x_1,\ldots,x_n)=\sum\limits_{j=1}^{n}\textrm{log}\, x_j,
\end{equation*}
\begin{equation*}
f(x_1,\ldots,x_n)=\bigl(\frac{1}{n}\sum\limits_{j=1}^{n}x_j^r\bigr)^{1/r},\ r<1,
\end{equation*}
also satisfy the conditions in Corollary \ref{corc}.

Finally, we present some inequalities between the symplectic eigenvalues and the usual 
eigenvalues of $A.$

\begin{theorem}\label{thm11}
Let $A \in \mathbb{P}(2n).$ Let $d_j(A), 1 \leq j \leq n$ be the symplectic 
eigenvalues of $A$ counted as in \eqref{eq2} and $\widehat{d}(A)$ the 
$2n$-tuple defined in \eqref{eq7}. Let $\lambda_1 (A), \lambda_2 (A), \ldots, 
\lambda_{2n} (A)$ be the usual eigenvalues of $A.$ Arranged in decreasing order 
they will be denoted by $\lambda_j^{\downarrow} (A)$ and in increasing order by 
$\lambda_j^{\uparrow} (A).$ Then
\begin{eqnarray}
{\rm (i)} && \widehat{d} (A) \prec_{\log} \,\lambda (A) 
\label{eq66} \\          
{\rm (ii)} && \lambda_j^{\uparrow}(A) \leq d_j(A) 
\leq \lambda_{n+j}^{\uparrow}(A), \quad 1 \leq j \leq n.       \label{eq67}
 \end{eqnarray}
 \end{theorem}

\begin{proof}
 \begin{itemize}
  \item[(i)] By the arguments seen in Section 3
$$\widehat{d_1} (A) = \| A^{1/2} J A^{1/2} \| \leq \| AJ \| = \|A\| = 
\lambda_1^{\downarrow} (A). $$
Arguing as before with $\Lambda^k A$ we get
$$\prod_{j=1}^{k} \widehat{d_j} (A) \leq \prod_{j=1}^{k} \lambda_j^{\downarrow} 
(A), \quad 1 \leq k \leq 2n.$$
When $k=2n$ both sides are equal to $\det\,A.$ This proves \eqref{eq66}.
\item[(ii)] It follows from the inequality $i J\le I$ that $ A^{1/2} i J  A^{1/2} \leq A.$ The eigenvalues of $ A^{1/2} i 
J  A^{1/2}$ arranged in increasing order are
$$- d_n (A) \leq \cdots \leq -d_1 (A) \leq d_1 (A) \leq \cdots \leq d_n(A),$$
and those of $A$ are
$$\lambda_1^{\uparrow}(A) \leq \cdots \leq \lambda_n^{\uparrow}(A) \leq  
\lambda_{n+1}^{\uparrow}(A) \leq \cdots \leq \lambda_{2n}^{\uparrow}(A).$$
By Weyl's monotonicity principle \cite{rbh},p.63
$$d_j(A) \leq \lambda_{n+j}^{\uparrow} (A) \qquad \mbox{for}\qquad 1 \leq j 
\leq n.$$
Replacing $A$ by $A^{-1}$ in this inequality we see that $\frac{1}{d_j(A)} \leq 
\frac{1}{\lambda_j^{\uparrow} (A)}$ for all $1 \leq j \leq n.$ This proves 
\eqref{eq67}.
 \end{itemize}
\end{proof}

\noindent{\bf{\it Caveat.}} In this paper we have chosen $J_{2n}=\begin{bmatrix}O & I\\
-I & O\end{bmatrix}$. Some authors choose instead $J_{2n}=J_2\oplus\cdots\oplus J_2\, $(n copies). Then the class of symplectic matrices, as well as the symplectic eigenvalues change. All our theorems remain valid with these changes.

\section*{Acknowledgements}
The authors thank Professor K. R. Parthasarathy and Dr. Ritabrata Sengupta for 
introducing them to this topic. The first author is supported by a J. C. Bose 
National Fellowship and the second author by a SERB Women's Excellence Award. 
The first author thanks Professor Qing-Wen Wang and the Department of 
Mathematics at Shanghai University for their warm hospitality in June 2015 
when a part of this work was done.

\end{document}